\documentclass[11pt]{extarticle}

\pdfoutput=1

\newcommand{\AUTHORS}{Kanat Tangwongsan, A. Pavan,  and Srikanta Tirthapura}
\newcommand{\TITLE}{Parallel Triangle Counting in Massive Streaming Graphs}
\newcommand{\KEYWORDS}{Put your keywords here}
\newcommand{\PAGENUMBERS}{yes}       
\newcommand{\COLOR}{yes}
\newcommand{\showComments}{yes}
\newcommand{\comment}[1]{}

\newcommand{\remove}[1]{}

\usepackage[T1]{fontenc}
\usepackage{ifthen}

\usepackage{ifthen}
\ifthenelse{\equal{\PAGENUMBERS}{yes}}{%
\usepackage[nohead,
            margin=0.85in,
            footskip=0.5in,bottom=1in,     
            columnsep=0.25in
            ]{geometry}
}{%
\usepackage[noheadfoot,left=0.75in,right=0.85in,top=0.75in,
            footskip=0.5in,bottom=1in,
            columnsep=0.25in
	    ]{geometry}
}

\let\chapter\section  

\usepackage[ruled,vlined]{algorithm2e}
\usepackage{float}
\floatstyle{ruled}
\newfloat{algo}{tbp}{lop}[section]
\floatname{algo}{{\small Algorithm}}
\newfloat{invr}{tbp}{lop}[section]
\floatname{invr}{{\small Invariant}}

\usepackage[bf,font=small,aboveskip=4pt]{caption}
\DeclareCaptionType{copyrightbox}

\usepackage{enumitem}
\setlist{itemsep=0pt,parsep=0pt}             

\usepackage[square, numbers]{natbib}

\usepackage{booktabs}
\usepackage[usenames,dvipsnames,svgnames,table]{xcolor}
\usepackage{colortbl}
\usepackage{multirow}

\usepackage{amsmath}
\usepackage{amssymb}
\usepackage{bm}
\usepackage{bbm}
\usepackage{ktmath}
\usepackage{xfrac}

\newenvironment{mlcodea}{
\begingroup
\ttfamily
\small
\begin{tabbing}
==\===\===\===\===\===\===\
\kill}
{\end{tabbing}
\endgroup}

\usepackage{txfonts}                            
\usepackage{txgreeks}

\DeclareSymbolFont{cmlargesymbols}{OMX}{cmex}{m}{n}  
\let\sum\relax
\DeclareMathSymbol{\sum}{\mathop}{cmlargesymbols}{"50}

\ifthenelse{\equal{\COLOR}{yes}}{%
  \usepackage[colorlinks,
    citecolor=NavyBlue,
    linkcolor=RedViolet,
    pagebackref=true
    ]{hyperref
  }
}{%
  \usepackage[pdfborder={0 0 0}]{hyperref}
}
\usepackage{url}

\hypersetup{%
pdfauthor = {\AUTHORS},
pdftitle = {\TITLE},
pdfkeywords = {\KEYWORDS},
bookmarksopen = {true}
}

\definecolor{placeholderbg}{rgb}{0.85,0.85,0.85}

\usepackage{graphicx}
\newcommand{\adjdeg}{\chi\xspace}
\newcommand{\defn}[1]{\textbf{#1}}
\newcommand{\MM}{M}
\newcommand{\BB}{B}
\newcommand{\CoQ}[1]{Q\left({#1}\right)}
\newcommand{\etal}{et~al.\xspace}

\newcommand{\Nbr}{\Gamma}
\newcommand{\stream}{\mathcal{S}}
\newcommand{\tcount}{\tau}
\newcommand{\tri}{\mathcal{T}}

\newcommand{\rank}{\mathop{\mathrm{rank}}}
\newcommand{\procName}[1]{{\normalfont \texttt{#1}}}
\newcommand{\scanAlg}{\procName{scan}}
\newcommand{\sortAlg}{\procName{sort}}

\newtheorem{observation}[theorem]{Observation}

\newcommand{\Qt}{\texttt{t}}
\newcommand{\Qb}{\texttt{b}}
\newcommand{\Qs}{\texttt{s}}
\newcommand{\Qc}{\texttt{c}}
\newcommand{\cc}{Q^{*}}
\newcommand{\loc}{\textit{loc}}

\renewcommand{\otilde}{\tilde{O}}

\newcommand{\tspc}{\mbox{\;\;\;}}

\usepackage[absolute]{textpos}

\newcommand{\note}[2]{
    \ifthenelse{\equal{\showComments}{yes}}{\textcolor{#1}{#2}}{}
}

\date{} \title{\TITLE\footnote{Contact Address: 1101 Kitchawan Rd, Yorktown
    Heights, NY 10598. E-mail: \texttt{ktangwo@us.ibm.com, pavan@cs.iastate.edu,
      snt@iastate.edu}}}

\author{{\large Kanat Tangwongsan$^\dag$ \qquad A. Pavan$^\ddag$ \qquad Srikanta Tirthapura$^\ddag$}\\[0.5em]
  {\it \mbox{}$^\dag$IBM T.J. Watson Research Center and \mbox{}$^\ddag$Iowa State University}}

\theoremstyle{definition}
\newtheorem{invariant}[theorem]{Invariant}


\usepackage{xspace}
\usepackage{balance}  
\usepackage{microtype}  

\begin{document}
\maketitle

\begin{abstract}
  The number of triangles in a graph is a fundamental metric, used in social
  network analysis, link classification and recommendation, and more.  Driven by
  these applications and the trend that modern graph datasets are both large and
  dynamic, we present the design and implementation of a fast and
  cache-efficient parallel algorithm for estimating the number of triangles in a
  massive undirected graph whose edges arrive as a stream.  It brings together
  the benefits of streaming algorithms and parallel algorithms.
  By building on the streaming algorithms framework, the algorithm has a small
  memory footprint.  By leveraging the paralell cache-oblivious framework, it
  makes efficient use of the memory hierarchy of modern multicore machines
  without needing to know its specific parameters. We prove theoretical bounds
  on accuracy, memory access cost, and parallel runtime complexity, as well as
  showing empirically that the algorithm yields accurate results and substantial
  speedups compared to an optimized sequential implementation.

  \medskip

  \emph{(This is an expanded version of a CIKM'13 paper of the same title.)}




\end{abstract}


\section{Introduction}
\label{sec:intro}

The number of triangles in a graph is an important metric in social network
analysis~\cite{WF94,Newman:siamreview03}, identifying thematic structures of
networks~\cite{EckmannE:pnas2002}, spam and fraud detection~\cite{BBCG08}, link
classification and recommendation~\cite{TsourakakisDMKFL:snam11}, among
others.
Driven by these applications and further fueled by the growing volume of graph
data, the research community has developed many efficient algorithms for
counting and approximating the number of triangles in massive graphs.  
There have been several streaming algorithms for triangle counting,
including~\cite{BKS02,JG05,BFLS07,MMPS11,KMSS12,PavanTTW:vldb2013,JhaSP:arxiv12}.
Although these algorithms are well-suited for handling graph evolution, they
cannot effectively utilize parallelism beyond the trivial ``embarrassingly
parallel'' implementation, leaving much to be desired in terms of performance
(we further discuss this below).  On the other hand, there have been a number of
parallel algorithms for triangle counting, such
as~\cite{SV11,Cohen:cse09}. While these algorithms excel at processing static
graphs using multiple processors, they cannot efficiently handle constantly
changing graphs.
Yet, despite indications that modern graph datasets are both massive and
dynamic, \emph{none of the existing algorithms can efficiently handle graph
  evolution and fully utilize parallelism at the same time.}


In this paper, we present a fast shared-memory parallel algorithm for
approximate triangle counting in a high-velocity streaming graph.
Using streaming techniques, it only has to store a small fraction of the graph
but is able to process the graph stream in a single pass.  Using
parallel-computing techniques, it was designed for modern multicore architecture
and can take full advantage of the memory hierarchy.  The algorithm is
versatile: It can be used to monitor a streaming graph where high throughput and
low memory consumption are desired. It can also be used to process large static
graphs, reaping the benefits of parallelism and small memory footprints.
As is standard, our algorithm provides a randomized relative-error approximation
to the number of triangles: given $\vareps, \delta \in [0, 1]$, the estimate
$\hat{X}$ is an $(\vareps, \delta)$-approximation of the true quantity $X$ if
$|\hat{X} - X| \leq \vareps X$ with probability at least $1 - \delta$.

We have used the following principles to guide our design and implementation:
First, we work in the limited-space streaming model, where only a fraction of
the graph can be retained in memory.  This allows the algorithm to process
graphs much larger than the available memory. Second, we adopt the \emph{bulk
  arrival} model.  Instead of processing a single edge at a time, our algorithm
processes edges one batch at a time. This allows for
processing the edges of a batch in parallel, which results in higher throughput.
Third but importantly, we optimize for the memory hierarchy. The memory system
of modern machines has unfortunately become highly sophisticated, consisting of
multiple levels of caches and layers of indirection. Navigating this complex
system, however, is necessary for a parallel implementation to be efficient.  To
this end, we strive to minimize cache misses. We design our algorithm in the
so-called cache-oblivious framework~\cite{FLPR99, BGS10}, allowing the algorithm
to make efficient use of the memory hierarchy but without knowing the specific
memory/cache parameters (e.g., cache size, line size).


\medskip
\noindent\textbf{Basic Parallelization:} 
Before delving further into our algorithm, we discuss two natural
parallelization schemes and why they fall short in terms of performance.
For the sake of this discussion, we will compare the PRAM-style (theoretical)
time to process a stream of $m$ total edges on $p$ processors, assuming a
reasonable scheduler.


As a concrete running example, we start with the sequential triangle counting
algorithm in~\cite{PavanTTW:vldb2013}, which uses the same neighborhood
sampling technique as our algorithm.  Their algorithm---and typically every
streaming algorithm for triangle counting---works by constructing a ``coarse''
estimator with the correct expectation, but with potentially high
variance. Multiple independent estimators are then run in parallel, and their
results are aggregated to get the desired $(\vareps, \delta)$-approximation.
This suggests a natural approach to parallelizing this algorithm (aka. the
``embarrassingly parallel'' approach): treat these independent estimators as
independent processes and update them simultaneously in parallel.  In this
\emph{na\"ive parallel} scheme, if there are $r$ estimators, the time on $p$
processors is $T_p = O(rm/p)$ and the total work is $O(rm)$ assuming each edge
update takes $O(1)$ time.  The problem is that the total work---the product
$r\cdot m$---can be prohibitively large even for medium-sized graphs, quickly
rendering the approach impractical even on multiple cores.

In the batch arrival model, there is an equally simple parallelization scheme
that bootstraps from a sequential batch processing algorithm.  Continuing with
our running example, the sequential algorithm in~\cite{PavanTTW:vldb2013} can
be adapted to bulk-process edges in a cache-efficient manner.  With the
modification, a batch of $s$ edges can be processed using $r=\Theta(s)$
estimators, in $O(s \log s)$ time per batch, assuming a cache-efficient
implementation.  In the second parallel scheme, which we term \emph{independent
  bulk parallel}, we $r/p$ estimators on each processor, and each processor uses
the bulk processing streaming algorithm to maintain the estimators assigned to
the processor. Since each processor has to handle all $m$ edge arrivals, this
leads to a parallel time of $T_p = \Theta(m\log s)$, and total work of
$\Theta(mp \log s)$.  Notice that in this scheme, the different parallel tasks
do not interact; they run independently in parallel.  While this appears to be
better than the na\"ive parallel approach, it suffers a serious drawback: the
total work increases linearly with the number of processors, dwarfing the
benefits of parallelization.

\subsection{Contributions}
The main contribution of this work is a shared-memory parallel algorithm, which
we call the \emph{coordinated bulk parallel} algorithm.
The key to its efficiency is the coordination between parallel tasks, so that we
do not perform duplicate work unnecessarily.
The algorithm maintains a number of estimators as chosen by the user and
processes edges in batches.  Upon the arrival of a batch, the algorithm (through
a task scheduler) enlists different cores to examine the batch and generate a
shared data structure in a coordinated fashion.  Following that, it instructs
the cores (again, via a task scheduler) to look at the shared data structure
and use the information they learn to update the estimators.
This is to be contrasted with the ``share-nothing'' schemes above.

We accomplish this by showing how to reduce the processing of a batch of edges
to simple parallel operations such as sorting, merging, parallel prefix, etc.
We prove that \emph{the processing cost of a batch of $s$ edges is
  asymptotically no more expensive than that of a cache-optimal parallel sort.}
(See Theorem~\ref{thm:pco-block-cost} for a precise statement.)  To put this in
perspective, this means that incorporating a batch of $s$ edges into $\Theta(s)$
estimators has a (theoretical) time of $T_p = \Theta(\log^2 s + \frac{s \log
  s}{p})$ assuming a reasonable scheduler.  Equivalently, using equal-sized
batches, the total work to process a stream of $m$ edges is $\Theta(m \log s)$.
This is $p$ times better than independent bulk parallel.  Furthermore, the
algorithm is cache oblivious, capable of efficiently utilizing of the caches
without needing to know the cache parameters


We also experimentally evaluate the algorithm.  Our implementation yields
substantial speedups on massive real-world networks.  On a machine with 12
cores, we obtain up to $11.24$x speedup when compared with a sequential version,
with the speedup ranging from $7.5$x to $11.24$x. In separate a stress test, a
large (synthetic) power-law graph of size 167GB was processed in about 1,000
seconds.  In terms of throughput, this translates to millions of edges per
second.


More broadly, our results show that it is possible to combine a small-space
streaming algorithm with the power of multicores to process massive dynamic
graphs of hundreds of GB on a modest machine with smaller memory. To our
knowledge, this is the first parallel small-space streaming algorithm for any
non-trivial graph property.

\subsection{Related Work}
\label{sec:related}
Approximate triangle counting is well-studied in both streaming and
non-streaming settings. In the streaming context, there has been a long line of
research~\cite{JG05,BFLS07,MMPS11,KMSS12}, beginning with the work of
Bar-Yossef~\etal\cite{BKS02}. Let $n$ be the number vertices, $m$ be the number
of edges, $\Delta$ be the maximum degree, and $\tau(G)$ be the number of
triangles in the graph.  An algorithm of~\cite{JG05} uses
$\otilde(m\Delta^2/\tau(G))$\footnote{The notation $\tilde{O}$ suppresses
  factors polynomial in $\log m, \log(1/\delta)$, and $1/\vareps$.} space
whereas the algorithm of \cite{BFLS07} uses $\otilde(mn/\tau(G))$ space. With
higher space complexity, \cite{MMPS11} and \cite{KMSS12} gave algorithms for the
more general problem of counting cliques and cycles, supporting insertion and
deletion of edges. A recent work~\cite{PavanTTW:vldb2013} presented an
algorithm with space complexity $\tilde{O}(m\Delta/\tau(G))$.  Jha
\etal~\cite{JhaSP:arxiv12} gave a $O(\sqrt{n})$ space approximation algorithm
for triangle counting as well as the closely related problem of the clustering
coefficient.  Their algorithm has an additive error guarantee as opposed to the
previously mentioned algorithms, which had relative error guarantees.  The
related problem of approximating the triangle count associated with each vertex
has also been studied in the streaming context~\cite{BBCG08,KutzkovP:wsdm2013},
and there are also multi-pass streaming algorithms for approximate triangle
counting~\cite{KolountzakisMPT:waw10}.  However, no non-trivial parallel
algorithms were known so far, other than the naive parallel algorithm.

In the non-streaming (batch) context there are many works on counting and
enumerating triangles---both exact and
approximate~\cite{CT85,Latapy08,SW05,TKMF09,BFNPW11,CC11}. Recent works on
parallel algorithms in the MapReduce model
include~\cite{SV11,Cohen:cse09,PaghT:ipl12}.



\section{Preliminaries and Notation}
\label{sec:prelim}
We consider a simple, undirected graph $G = (V, E)$ with vertex set
$V$ and edge set $E$. The edges of $G$ arrive as a stream, and we
assume that every edge arrives exactly once. Let $m$ denote the number
of edges and $\Delta$ the maximum degree of a vertex in $G$. 
An edge $e \in E$ is a size-$2$ set consisting of its endpoints. 
For edges $e, f \in E$, we say that $e$ is \defn{adjacent to} $f$
or $e$ is \defn{incident on} $f$ if they share a vertex---i.e., 
$|e \cap f| = 1$.
When the graph $G$ has a total order (e.g., imposed by the stream arrival
order), we denote by $\stream = (V, E, \leq_\stream)$ the graph $G = (V,E)$,
together with a total order $\leq_\stream$ on $E$.  The total order fully
defines the standard relations (e.g., $<_\stream, >_\stream, \geq_\stream$),
which we use without explicitly defining.  When the context is clear, we
sometimes drop the subscript.  Further, for a sequence $A = \langle a_1, \dots,
a_{|A|} \rangle$, we write $G_A = (V_A, E_A, \leq_A)$, where $V_A$ is the
relevant vertex set, $E_A = \{a_1, \dots, a_{|A|}\}$, and $\leq_A$ is the total
order defined by the sequence order.
Given $\stream = (V, E, \leq_\stream)$, the \defn{neighborhood of an edge $e \in
  E$}, denoted by $\Nbr_\stream(e)$, is the set of all edges in $E$ incident on
$e$ that ``appear after'' $e$ in the $\leq_\stream$ order; that is,
$\Nbr_\stream(e) := \{ f \in E : f \cap e \neq \emptyset \land f >_\stream e\}$.

Let $\tri(G)$ (or $\tri(\stream)$) denote the set of all triangles in
$G$---i.e., the set of all closed triplets, and $\tau(G)$ be the number of
triangles in $G$.  For a triangle $t^* \in \tri(G)$, define $C(t^*)$ to be
$|\Nbr_\stream(f)|$, where $f$ is the smallest edge of $t^*$
w.r.t. $\leq_\stream$.  Finally, we write $x \in_R S$ to indicate that $x$ is a
random sample from $S$ taken uniformly at random.

\medskip
\noindent
\textbf{Parallel Cost Model.}  We focus on parallel algorithms that efficiently
utilize the memory hierarchy, striving to minimize cache\footnote{The term
  \defn{cache} is used as a generic reference to a level in the memory
  hierarchy; it could be an actual cache level (L1, L2, L3), the TLB, or page
  faults)} misses across the hierarchy.
Towards this goal, we analyze the (theoretical) efficiency of our algorithms in
terms of the memory cost---i.e., the number of cache misses---in addition to the
standard parallel cost analysis.

The specifics of these notions are not necessary to understand this paper
(for completeness, more details appear in Appendix~\ref{sec:wd-pco}). In the
remainder of this section, we summarize the basic ideas of these concepts: The
\defn{work} measure counts the total operations an algorithm performs.
The \defn{depth} measure is the length of the longest chain of dependent tasks
in the algorithm; this is a lower bound on the parallel runtime of the
algorithm. A gold standard here is to keep work around that of the best
sequential algorithm and to achieve small depth (sublinear or polylogarithmic in
the input size).

For memory cost, we adopt the \emph{parallel cache-oblivious (PCO)}
model~\cite{BFGS11}, a well-accepted parallel variant of the
cache-oblivious model~\cite{FLPR99}. A cache-oblivious algorithm
has the advantage of being able to make efficient use of the memory
hierarchy without knowing the specific cache parameters (e.g. cache
size, line size)---in this sense, the algorithm is oblivious to the
cache parameters. This makes the approach versatile, facilitating
code reuse and reducing the number of parameters that need
fine-tuning.
%
In the PCO model, the cache/memory complexity of an algorithm $A$ is given as a
function of cache size $\MM{}$ and line size $\BB{}$ assuming the optimal
offline replacement policy\footnote{In reality, practical heuristics such
least-recently used (LRU) are used and are known to have competitive
performance with the in-hindsight optimal policy.}. 
This \defn{cache complexity measure} is denoted by $\cc(A; \MM{}, \BB{})$, which
behaves much like work measure and subsumes it for the applications in this
paper.
In the context of a parallel machine, it represents the number of cache misses
across all processors for a particular level.  Furthermore, because the
algorithm is oblivious to the cache parameters, the bounds simultaneously hold
across all levels in the memory hierarchy.  
\section{Technical Background}
\label{sec:tech-bk}

\subsection{Neighborhood Sampling}
\label{sec:nbr-sampling}
We review \emph{neighborhood sampling}, a technique for selecting a random
triangle from a streaming graph. This was implicit in the streaming algorithm in
a recent work~\cite{PavanTTW:vldb2013}.  When compared to other sampling
algorithms proposed in the context of streaming graphs, neighborhood sampling
has a higher probability of discovering a triangle, which translates to a better
space bound in the context of triangle counting.
For this paper, we will restate it as a set of invariants:
\begin{invariant}
  \label{invariant:nbsi}
  Let $\stream = (V, E, \leq_\stream)$ denote a simple, undirected graph $G =
  (V,E)$, together with a total order $\leq_\stream$ on $E$.  The tuple $(f_1,
  f_2, f_3, \adjdeg)$, where $f_i \in E \cup \{\emptyset\}$ and $\adjdeg \in
  \Z_+$, satisfies the \defn{neighborhood sampling invariant (NBSI)} if
  \begin{enumerate}[label=(\arabic*),topsep=1pt,itemsep=1pt]
  \item \textbf{Level-1 Edge:} $f_1 \in_R E$ is chosen uniformly at random from $E$;
  \item $\adjdeg = |\Nbr_{\stream}(f_1)|$ is the number of edges in $\stream$
    incident on $f_1$ that appear after $f_1$ according to $\leq_\stream$.
  \item \textbf{Level-2 Edge:} $f_2 \in_R \Gamma_{\stream}(f_1)$ is chosen uniformly from the
    neighbors of $f_1$ that appear after it (or $\emptyset$ if the neighborhood
    is empty); and
  \item \textbf{Closing Edge:} $f_3 >_\stream f_2$ is an edge that completes the triangle
  formed uniquely defined by the edges $f_1$ and $f_2$ (or
    $\emptyset$ if the closing edge is not present).
  \end{enumerate}
\end{invariant}

The invariant provides a way to maintain a random---although
non-uniform---triangle in a graph.  We state a lemma that shows how to obtain an
unbiased estimator for $\tau$ from this invariant (the proof of this lemma
appears in~\cite{PavanTTW:vldb2013} and is reproduced here for completeness):
\begin{lemma}
 \label{lem:unbiased-est}
 Let $\stream = (V, E, \leq_\stream)$ denote a simple, undirected graph $G =
 (V,E)$, together with a total order $\leq_\stream$ on $E$.  Further, let $(f_1,
 f_2, f_3, \adjdeg)$ be a tuple satisfying NBSI. Define random variable $X$ as
 $X = 0$ if $f_3= \emptyset$ and $X=\adjdeg \cdot |E|$ otherwise.  Then,
 $\expct{X} = \tcount(G)$.
\end{lemma}

The lemma follows directly the following claim, which establishes the
probability that the edges $f_1, f_2, f_3$ coincide with a particular triangle
in the graph $G$.
\begin{claim}
  \label{claim:discovery-prob}
  Let $t^*$ be any triangle in $G$.  If $(f_1, f_2, f_3, \adjdeg)$ satisfies
  NBSI, then the probability that $\{f_1, f_2, f_3\}$ represents the triangle
  $t^*$ is
  \[
  \prob{\{f_1, f_2, f_3\} = t^*} = \frac{1}{|E| \cdot C_\stream(t^*)}
  \]
  where we recall that $C(t^*) = |\Nbr_\stream(f)|$ if $f$ is the $t^*$'s first
  edge in the $\leq_\stream$ ordering.
\end{claim}
\begin{proof}
  Let $t^* = \{e_1, e_2, e_3\} \in \tri(G)$, where $e_1 < e_2 < e_3$ without
  loss of generality. Further, let $\event_1$ be the event that $f_1 = e_1$, and
  $\event_2$ be the event that $f_2 = e_2$.  It is easy to check that $f_3 =
  e_3$ if and only if both $\event_1$ and $\event_2$ hold.  By NBSI, we have
  that $\prob{\event_1} = \frac{1}{|E|}$ and $\prob{\event_2 \/\mid\/ \event_1}
  = \frac{1}{|\Nbr_{\stream}(e_1)|} = \frac{1}{C_\stream(t^*)}$.  Thus, $\prob{t
    = t^*} = \prob{\event_1 \cap \event_2} = \frac{1}{|E|\cdot C_\stream(t^*)}$,
  concluding the proof.
\end{proof}



While the estimate provided by Lemma~\ref{lem:unbiased-est} is unbiased, it has
a high variance. We can obtain a sharper estimate by running multiple copies of
the estimate (making independent random choices) and aggregating them, for
example, using median-of-means aggregate.  The proof of the following theorem is
standard; interested readers are referred to~\cite{PavanTTW:vldb2013}:
%
\begin{theorem}[\cite{PavanTTW:vldb2013}]
  \label{thm:mdelta-suffices}
  There is an $(\vareps, \delta)$-approximation to the triangle counting problem
  that requires at most $r$ independent estimators on input a graph $G$ with $m$
  edges, provided that $r \geq \frac{96}{\vareps^2} \cdot
  \frac{m\Delta(G)}{\tau(G)} \cdot \log (\tfrac1\delta)$.
\end{theorem}

\subsection{Parallel Primitives}
Instead of designing a parallel cache-oblivious algorithm directly, we describe
our algorithm in terms of primitive operations that have parallel algorithms
with optimal cache complexity and polylogarithmic depth.  This not only
simplifies the exposition but also allows us to utilize existing implementations
that have been optimized over time.

Our algorithm relies on the following primitives: \procName{sort},
\procName{merge}, \procName{concat}, \procName{map}, \procName{scan},
\procName{extract}, and \procName{combine}.  
The primitive \procName{sort} takes a sequence and a comparison function, and
outputs a sorted sequence. The primitive \procName{merge} combines two sorted
sequences into a new sorted sequence.  The primitive \procName{concat} forms a
new sequence by concatenating the input sequences.
The primitive \procName{map} takes a sequence $A$ and a function $f$,
and it applies $f$ on each entry of $A$.  The primitive \procName{scan}
(aka. prefix sum or parallel prefix) takes a sequence $A$ ($A_i \in D$), an
associative binary operator $\oplus$ ($\oplus\!\!: D\times D \to D$), a
left-identity $\textsf{id}$ for $\oplus$ ($\textsf{id} \in D$), and it produces
the sequence $\langle \textsf{id}, \textsf{id} \oplus A_1, \textsf{id} \oplus
A_1 \oplus A_2, \dots, \textsf{id} \oplus A_1 \oplus \dots \oplus A_{|A|-1}
\rangle$.  The primitive \procName{extract} takes two sequences $A$ and $B$,
where $B_i \in \{1, \dots, |A|\} \cup \{\procName{null}\}$, and returns a
sequence $C$ of length $|B|$, where $C_i = A[B_i]$ or \procName{null} if $B_i =
\procName{null}$.  Finally, the primitive \procName{combine} takes two sequences
$A, B$ of equal length and a function $f$, and outputs a sequence $C$ of length
$|A|$, where $C[i] = f(A[i], B[i])$.

On input of length $N$, the cache complexity of sorting in the PCO
model\footnote{As is standard, we make a technical assumption known as the
  tall-cache assumption (i.e., $\MM \geq \Omega(\BB^2)$).} is $\cc(\sortAlg(N);
\MM,\BB) = O(\frac{N}{B} \log_{M/B} (1+\frac{N}{B}))$, whereas
\procName{concat}, \procName{map}, \procName{scan}, and \procName{combine} all
have the same cost as $\procName{scan}$: $\cc(\scanAlg(N); \MM, \BB) =
O(N/\BB)$. For \procName{merge} and \procName{concat}, $N$ is the length of the
two sequences combined. We also write $\sortAlg(N)$ and $\scanAlg(N)$ to denote
the corresponding cache costs when the context is clear.  In this notation, the
primitive \procName{extract} has $O(\sortAlg(|B|) + \scanAlg(|A| + |B|))$.  All
these primitives have at most $O(\log^2 N)$ depth.


In addition, we will rely on a primitive for looking up multiple keys from a
sequence of key-value pairs.  Specifically, let $S = \langle (k_1, v_1), \dots,
(k_n, v_n) \rangle$ be a sequence of $n$ key-value pairs, where $k_i$ belongs to
a total order domain of keys. Also, let $T = \langle k'_1, \dots, k'_m \rangle$
be a sequence of $m$ keys from the same domain.  The \emph{exact multisearch
  problem} (\procName{exactMultiSearch}) is to find for each $k'_j \in T$ the
matching $(k_i, v_i) \in S$. We will also use the \emph{predecessor multisearch}
(\procName{predEQMultiSearch}) variant, which asks for the pair with the largest
key no larger than the given key.  The exact search version has a simple hash
table implementation that will not be cache friendly.  Existing cache-optimal
\sortAlg{} and \procName{merge} routines directly imply an implementation with
$O(\sortAlg(n) + \sortAlg(m) + \scanAlg(n+m))$ cost:

\begin{lemma}[\cite{BGS10,BFGS11}]
  \label{lem:ms-cost}
  There are algorithms $\procName{exactMultiSearch}(S, T)$ and
  $\procName{predEQMultiSearch}(S, T)$ each running in $O(\log^2 (n+m))$ depth
  and $O(\sortAlg(n) + \sortAlg(m))$ cache complexity, where $n = |S|$ and $m =
  |T|$.  Furthermore, if $S$ and $T$ have been presorted in the key order, these
  algorithms take $O(\log (n+m))$ depth and $O(\scanAlg(n+m))$.
\end{lemma}



\section{Parallel Streaming Algorithm}
\label{sec:par-tri-count}

In this section, we describe the coordinated parallel bulk-processing algorithm.
The main result of this section is as follows:

\begin{theorem}
  \label{thm:pco-block-cost}
  Let $r$ be the number of estimators maintained for approximate
  triangle counting. There is a parallel algorithm
  \procName{bulkUpdateAll} for processing a batch of $s$ edges with
  $O(\sortAlg(r) + \sortAlg(s))$ cache complexity (memory-access cost)
  and $O(\log^2 (r+s))$ depth.
\end{theorem}
This translates to a $p$-processor time of $T_p = \Theta((r\log r + s\log s)/p +
\log^2 (r+s))$ in a PRAM-type model.  To meet these bounds, we cannot afford to
explicitly track each estimator individually.  Neither can we afford a large
number of random accesses.

We present an overview of the algorithm and proceed to discuss its details and
complexity analysis.

\newcommand{\est}{\text{\normalfont \textit{est}}}
\newcommand{\dsto}{\text{\normalfont \textit{ds}}}

\subsection{Algorithm Overview}
We outline a parallel algorithm for efficiently maintaining $r$ independent
estimators satisfying NBSI as a bulk of edges arrives.  More precisely, for
$i= 1,\dots, r$, let $\est_i$ be a tuple $(f_1, f_2, f_3, \adjdeg)$ satisfying NBSI
on the graph $G = (V, E, \leq_{\stream})$, where $\leq_\stream$ gives a total
order on $E$. The sequence of arriving edges is modeled as $W = \langle w_1,
\dots, w_{s} \rangle$; the sequence order defines a total order on $W$.  Denote
by $G' = (V', E', \leq_{\stream'})$ the graph on $E \cup W$, where the edges of
$W$ all come after the edges of $E$ in the new total order $\stream'$.

The goal of the \procName{bulkUpdateAll} algorithm is to take as input
estimators $\est_i$ ($i = 1, \dots, r$) that satisfy NBSI on $G$ and the
arriving edges $W$, and to produce as output estimators $\est'_i$ ($i = 1,\dots,
r$) that satisfy NBSI on $G'$.  The algorithm has $3$ steps corresponding to the
main parts of NBSI. After these steps, the NBSI invariant is upheld; if the
application so chooses, it can aggregate the ``coarse'' estimates together in
cost no more expensive that the update process itself.
\begin{enumerate}[topsep=1.5pt,itemsep=2pt,label=\textbf{Step~\arabic*:},leftmargin=3\parindent]
\item Update level-1 edges;
\item Update level-2 edges and neighborhood sizes $\adjdeg$'s;
\item Check for closing edges.
\end{enumerate}
For each step, we will describe what needs to done conceptually for each
individual estimator and explain how to do it efficiently across all $r$
estimators in parallel. Slightly abusing notation to reduce clutter, we will
write $\Nbr_A(f)$ to mean $A \cap \Nbr_{\stream'}(f)$, where $A \subseteq
\stream'$ is a set and $f \in \stream'$ is an edge.





\subsection{Step 1: Manage Level-1 Edges}
The goal of Step 1 is to make sure that for each estimator, its level-1 edge
$f_1$ is a uniform sample of the edges that have arrived so far ($E \cup W$).
We begin by describing a conceptual algorithm for \emph{one} estimator.  This
step is relatively straightforward: use a simple variant of reservoir
sampling---with probability $\frac{|W|}{|W|+|E|}$, replace $f_1$ with an edge
uniformly chosen from $W$; otherwise, retain the current edge.  Notice that for
a batch size of $1$ (i.e., $|W| = 1$), this recovers the familiar reservoir
sampling algorithm.

Implementing this step in parallel is easy. First, we generate a length-$r$
vector $\procName{idx}$, where $\procName{idx}[i]$ is the index of the
replacement edge in $W$ or $\procName{null}$ if the current edge is retained.
This can be done by using the primitive \procName{map} on the function
\procName{randInt}$(|W| + |E|)$.  Then, we apply \procName{extract} and
\procName{combine} to ``extract'' these indices from $W$ and update the level-1
edges in the estimator states.  For the estimators receiving a new level-1 edge,
set its $\adjdeg$ to $0$ (this helps simplify the presentation in the next
step).

\subsection{Step 2: Update Level-2 Edges and Degrees}
The goal of Step 2 is to ensure that for every estimator $\est_i$, the level-2
edge $\est_i.f_2$ is chosen uniformly at random from
$\Nbr_{\stream'}(\est_i.f_1)$ and $\est_i.\adjdeg =
|\Nbr_{\stream'}(\est_i.f_1)|$. We remember that $\Nbr_{\stream'}(f_1) = \{ e
>_{\stream'} f_1 : e \cap f_1 \neq \emptyset \}$---or, in words, the set of
edges incident on $f_1$ that appear after it in $\stream'$. In this view,
$\adjdeg$ is the size of $\Nbr_{\stream'}(f_1)$, and $f_2$ is a random sample
from an appropriate ``substream.''

We now describe a conceptual algorithm for \emph{one} estimator. Consider an
estimator $\est_i = (f_1, f_2, f_3, \adjdeg)$ that has completed Step 1. To
streamline the presentation, we define two quantities:
\begin{equation*}
  \adjdeg^{-} = |\Nbr_{E}(f_1)|  \quad\text{ and }\quad  \adjdeg^{+} = |\Nbr_{W}(f_1)|
\end{equation*}
In words, $\adjdeg^{-}$ is the number edges in $E$ incident on $f_1$ that
arrived after $f_1$---and $\adjdeg^{+}$ is the number edges in $W$ incident on
$f_1$ that arrived after $f_1$.  Thus, $\adjdeg^{-} = \adjdeg$ (inheriting it
from the current state); remember that if $f_1$ was just replaced in Step 1,
$\adjdeg$ was reset to $0$ in Step 1.  We also note that
$|\Nbr_{\stream'}(f_1)|$---the total number of edges incident on $f_1$ that
arrived after it in the whole stream---is $\adjdeg^{-} + \adjdeg^{+}$.

In this notation, the update rule is simply:
\begin{quote}
  With probability $\frac{\adjdeg^{-}}{\adjdeg^{-} + \adjdeg^{+}}$, keep the
  current $f_2$; otherwise, pick a random edge from $\Nbr_W(f_1)$. Then, update
  $\adjdeg$ to $\adjdeg^{-} + \adjdeg^{+}$.
\end{quote}

\medskip
\noindent\textbf{Designing an efficient algorithm.} Before we can implement this
step efficiently, we have to answer the question:
\emph{How to efficiently compute, for every estimator, the number of candidate
edges $\adjdeg^{+}$ and how to sample uniformly from these candidates?}  This
is challenging because at a high level, all $r$ estimators need to navigate
their substreams---potentially all different---to figure out their sizes and
sample from them. To meet our performance requirement, we cannot afford to
explicitly track these substreams.  Moreover, for reasons of cache efficiency,
we can neither afford a large number of random accesses, though this
seems necessary at first glance.

We address this challenge in two steps: First, we define the notion of
\emph{rank} and present a fast algorithm for it. 
Second, we show how to sample efficiently using the {\em rank},
and how it relates to the number of potential candidates.

\begin{definition}[Rank]
  Let $W = \langle w_1, \dots, w_{s}\rangle$ be a sequence of unique edges. Let
  $H = (V_W, W)$ be the graph on the edges $W$, where $V_W$ is the set of
  relevant vertices on $E \cup W$.  For $x,y \in V_W$, $x\neq y$, the
  \emph{rank} of $x \to y$ is
  \[
  \rank(x \to y) = \begin{cases}
    |\{j :  x \in w_j \land j > i \}| & \text{ if }\exists i, \{x, y\} = w_i\\
    \deg_{H}(x) & \text{ otherwise }
  \end{cases}
  \]
\end{definition}

\begin{figure}[tb]
  \centering
  \begin{minipage}{1.85in}
    \includegraphics[width=1.85in]{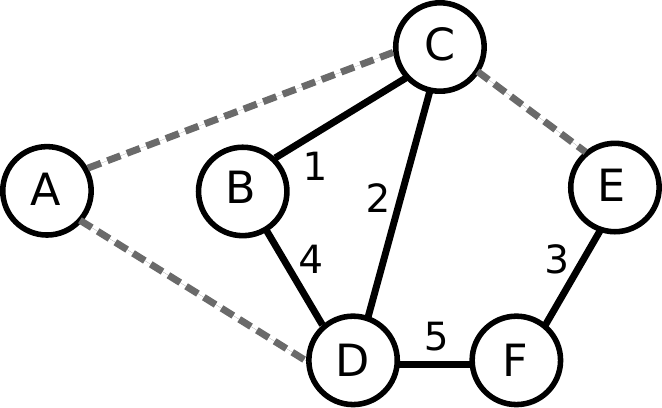}
  \end{minipage}
  \;\,\;
  \begin{minipage}{1in}
    {
    \begin{tabular}{l c r}
      \toprule
      \small
      Arc & \small Rank\\
      \midrule
      \small {\sf C} $\to$ {\sf A} & \small 2\\
      \small \textsf{C} $\to$ \small \textsf{D} & \small 0\\
      \small \textsf{D} $\to$ \small \textsf{C} & \small 2\\
      \small \textsf{D} $\to$ \small \textsf{F} & \small 0\\
      \small \textsf{B} $\to$ \small \textsf{C} & \small 1\\
      \small \textsf{B} $\to$ \small \textsf{D} & \small 0\\
      \bottomrule
    \end{tabular}
  }
  \end{minipage}
  \smallskip
  \caption{A $6$-node streaming graph where a batch of $5$ new edges (solid;
    arrival order labeled next to them) is being added to a stream of $3$ edges
    (dashed) that have arrived earlier---and examples of the corresponding
    \textrm{rank} values as this batch is being incoporated.}
  \label{fig:adding-example-1}
\end{figure}


In words, if $\{x, y\}$ is an edge in $W$, $\rank(x \to y)$ is the number of
edges in $W$ that are incident on $x$ and appear after $xy$ in $W$. For non-edge
pairs, $\rank(x \to y)$ is simply the degree of $x$ in the graph $G_W$.  This
function is, in general, not symmetric: $\rank(x \to y)$ is not the same as
$\rank(y \to x)$.  We provide an example in Figure~\ref{fig:adding-example-1}.


%
\medskip
\noindent\textbf{Computing rank.}  We give an algorithm for computing the rank
of every edge, in both orientations; it outputs a sequence of length $2|W|$ in
an order convenient for subsequent lookups.  Following the description and
proof, we show a run of the algorithm on an example graph.
\begin{lemma}
  \label{lem:rank-all}
  There is a parallel algorithm $\procName{rankAll}(W)$ that takes a sequence of
  edges $W = \langle w_1, \dots, w_{s} \rangle$ and produces a sequence of
  length $2|W|$, where each entry is a record $\{ \procName{src},
  \procName{dst}, \procName{rank}, \procName{pos} \}$ such that
  \begin{enumerate}[topsep=0pt]
  \item $\{\procName{src}, \procName{dst}\} = w_i$ for some $i = 1, \dots, |W|$;
  \item $\procName{pos} = i$; and
  \item $\procName{rank} = \rank(\procName{src} \to
    \procName{dst})$.
  \end{enumerate}
  Each input edge $w_i$ gives rise to exactly $2$ entries, one per orientation.
  The algorithm runs in $O(\sortAlg(s))$ cache complexity and $O(\log^2 s)$
  depth.
\end{lemma}
(Readers interested in the main idea of the algorithm may wish to skip this
proof on the first read and proceed to the example in
Figure~\ref{fig:rank-example}.)
\begin{proof}
  We describe an algorithm and reason about its complexity.  First, we form a
  sequence $F$, where each $w_i = \{u, v\}$ yields two records:
  $\{\procName{src} = u, \procName{dst} = v, \procName{pos} = i\}$ and
  $\{\procName{src} = v, \procName{dst} = u, \procName{pos} = i\}$.  That is,
  there are two entries corresponding to each input edge, both marked with the
  position in the original sequence. This step can be accomplished in
  $O(\scanAlg(s))$ cache complexity and at most $O(\log^2 s)$ depth using
  \procName{map} and \procName{concat}.

  If we view the entries of $F$ as directed edges, the rank of an arc $e \in F$,
  $\rank(e.\procName{src} \to e.\procName{dst})$, is the number of arcs in $F$
  emanating from $e.\procName{src}$ with $\procName{pos} > e.\procName{pos}$. To
  take advantage of this equivalent definition, we sort $F$ by \procName{src}
  and for two entries of the same \procName{src}, order them in the decreasing
  order of \procName{pos}.  This has $O(\sortAlg(s))$ cost and $O(\log^2 s)$
  depth.  Now, in this ordering, the rank of a particular edge is one more than
  the rank of the edge immediately before it unless it is the first edge of that
  \procName{src}; the latter has rank $0$. Consequently, the rank computation
  for the entire sequence has cost $O(\scanAlg(s))$ since all that is required
  is to figure out whether an edge is the first edge with that \procName{src}
  (using \procName{combine}) and a prefix scan operation (see, e.g.,
  \cite{Jaja:book92} or Appendix~\ref{sec:scan-with-resets}).
  Overall, the algorithm runs in $O(\sortAlg(s) + \scanAlg(s))$ cache complexity
  and $O(\log^2 s)$ depth, as claimed.
\end{proof}
\begin{figure}[t]
  \procName{rankAll}$(W) = \;$\\
  \begin{center}
    \vspace{-0.4cm}
  {\sffamily
  \begin{tabular}{r r c c c c c c c c c c}
    \toprule
    \multirow{3}{*}{\textbf{\fbox{I}}}
    &\footnotesize src: &
    B & B & C &C & D &D &D &E & F &F\\
    &\footnotesize \textsf{dst:} &
    D & C & D & B & F & B & C & F & D & E\\
    &\footnotesize \textsf{pos:} &
    4 & 1 & 2 & 1 & 5 & 4 & 2 & 3 & 5 & 3 \\
    \midrule
    \textbf{\fbox{II}} &
    \footnotesize \textsf{rank:} &
    0 & 1 & 0 & 1 & 0 & 1 & 2 & 0 & 0 & 1\\
    \bottomrule
  \end{tabular}
}
\end{center}
\caption{Running \procName{rankAll} on the example streaming graph. Each column
  corresponds to an entry in the output.}
\label{fig:rank-example}
\end{figure}

In Figure~\ref{fig:rank-example}, we give an example of the output of
\procName{rankAll} on the example graph from Figure~\ref{fig:adding-example-1}.
It also illustrates the intermediate steps of the algorithm: the 3 rows marked
with $\textbf{\textsf{I}}$ show the array $F$ the algorithm has after the sort
step. The line marked with $\textbf{\textsf{II}}$ is the result after performing
\procName{scan}.

In addition, using this figure as an example, we make two observations about the
output of \procName{rankAll} that will prove useful later on.
\begin{enumerate}[topsep=1.5pt, itemsep=0pt]
\item it is ordered by \procName{src} breaking tie in the decreasing
  order of \procName{pos}; and
\item it is ordered by $\procName{src}$ then by
  $\procName{rank}$, in the increasing order.
\end{enumerate}

\medskip
\noindent\textbf{Mapping rank to substreams.}  Having defined rank and showed how
to compute it, we now present an easy-to-verify observation that
\emph{implicitly} defines a substream with respect to a level-1 edge $f_1$ in
terms of rank values:
%
\begin{observation}
  Let $f_1 = \{u, v\} \in E \cup W$ be a level-1 edge. Let $F' =
  \procName{rankAll}(W)$. The set of edges in $W$ incident on $f_1$ that appears
  after it---i.e., the set $\Nbr_W(f_1)$---is precisely the undirected edges
  corresponding to $L \cup R$, where
  \begin{align*}
    L &= \{ e \in F' : e.\procName{src} = u, e.\procName{rank} < \rank(u \to v) \}\\
    \intertext{and}
    R &= \{ e \in F' : e.\procName{src} = v, e.\procName{rank} < \rank(v \to u) \}.
  \end{align*}
\end{observation}

To elaborate further, we note that there are $\rank(u \to v)$ edges in
$\Nbr_W(f_1)$ incident on $u$ (those in $L$) and $\rank(v \to u)$ edges, on $v$
(those in $R$).  As an immediate consequence, we have that $\adjdeg^{+} =
\rank(u \to v) + \rank(v \to u)$.  Since $\adjdeg^{-}$ is already known, we can
readily compute the new $\adjdeg$ value as $\adjdeg = \adjdeg^{+} +
\adjdeg^{-}$, as noted before.

Furthermore, we will use this as a ``naming system'' for identifying which
level-2 edge to pick. For a level-1 edge $f_1$, there are $\adjdeg^{+}$
candidates in $\Nbr_W(f_1)$.  We will associate each number $\phi \in \{0, 1,
\dots, \adjdeg^{+}-1\}$ with an edge in $\Nbr_W(f_1)$ as follows: if $\phi <
\rank(u \to v)$, we associate it with the edge $\procName{src} = u$ and
$\procName{rank} = \phi$; otherwise, associate it with the edge $\procName{src}
= v$ and $\procName{rank} = \phi - \rank(u \to v)$.
We give two examples in Figure~\ref{fig:example-rank-to-substream} to illustrate
the naming system.

\begin{figure}
  \centering
  {\sffamily
  \begin{tabular}{ r c c || c }
    \toprule
    \footnotesize $\phi$:& 0 & 1 & \\
    \footnotesize src: & D & D & \\
    \footnotesize \textsf{rank:} & 0 & 1 & \\
    \midrule
    \footnotesize real edge:& DF & BD & \\
    \bottomrule
  \end{tabular}
}
  {\sffamily
  \begin{tabular}{ r c c || c }
    \toprule
    \footnotesize $\phi$:& 0 & 1 & 2 \\
    \footnotesize src: & C & C & E\\
    \footnotesize \textsf{rank:} & 0 & 1 & 0\\
    \midrule
    \footnotesize real edge:& CD & CB & EF\\
    \bottomrule
  \end{tabular}
}
\caption{Translating rank values to substreams: (Left) $f_1 = \textsf{\small
    DC}$ using $u = \textsf{\small D}$ and $v = \textsf{\small C}$, so
  $\adjdeg^{+} = 2$. (Right) $f_1 = \textsf{\small CE}$ using $u =
  \textsf{\small C}$ and $v = \textsf{\small E}$, so $\adjdeg^{+} = 3$.}
\label{fig:example-rank-to-substream}
\end{figure}


\medskip
\noindent\textbf{Putting them together.}
Armed with these ingredients, we are ready to sketch an algorithm for
maintaining level-2 edges and $\adjdeg$.  We will also rely on multisearch
queries of the following forms (which can be supported by
Lemma~\ref{lem:ms-cost})---\emph{(Q1)} Given $(u, p)$, locate an edge with
$\procName{src} = u$ with the least \procName{pos} larger or equal to $p$; and
\emph{(Q2)} Given $(u, a)$, locate an edge with $\procName{src} = u$ with the
rank exactly equal to $a$.

First, construct a length-$r$ array $\adjdeg^{-}[]$ of $\adjdeg^{-}$ values by
copying from $\est_i.\adjdeg$ (using \procName{map}). Then, for $i = 1, \dots,
r$, compute $\mathit{ld}[i] = \rank(u \to v)$ and $\mathit{rd}[i] = \rank(v \to
u)$, where $\{u, v\} = \est_i.f_1$; this can be done using two multisearch calls
each involving $r$ queries of the form \emph{(Q1)}\footnote{As an optimization
  to save time, for estimators whose level-1 edge did not get replaced, we query
  for $p = -1$; this will turn up the edge with the largest rank incident on
  that \procName{src}.}  This is sufficient information to compute the
length-$r$ array $\adjdeg^{+}[]$, where $\adjdeg^{+}[i] = \mathit{ld}[i] +
\mathit{rd}[i]$ (using \procName{combine}).

Now we apply \procName{map} to flip a coin for each estimator deciding whether
it will take on a new level-2 edge from $W$.  Using the same \procName{map}
operation, for every estimator $\est_i$ that is replacing its level-2 edge, we
pick a random number between $0$ and $\adjdeg^{+}[i] - 1$ (inclusive), which
maps to an edge in the substream $\Nbr_W(\est_i.f_1)$ using the naming scheme
above. To convert this number to an actual edge, we perform a multisearch
operation involving at most $r$ queries of the form \emph{(Q2)}, which completes
Step 2.

\subsection{Step 3: Locate  Closing Edges}
The goal of Step 3 is to detect, for each estimator, if the wedge formed by
level-1 and level-2 edges closes using a edge from $W$.
For this final step, we construct a length-$|W|$ array \procName{edges}, where
each $w_i = \{u, v\} \in W$ ($u < v$) yields a record $\{\procName{src} = u,
\procName{dst} = v, \procName{pos} = i\}$. We then sort this by
$\procName{src}$, then by $\procName{dst}$.  After that, with a \procName{map}
operation, we compute the candidate closing edge and use a multisearch on
$\procName{edges}$ with at most $r$ queries of the form ``Given $(u, v)$, locate
the edge $(u, v)$'' This allows us to check if the candidates are present and
come after the level-2 edge, by checking their \procName{pos} field.

\subsection{Cost Analysis}
Using the terminology of Theorem~\ref{thm:pco-block-cost}, we let $s = |W|$ and
$r$ be the number of estimators we maintain. Step 1 is implemented using
\procName{map}, \procName{extract}, and \procName{combine}. 
Therefore, the cache cost for Step 1 is at most $O(\sortAlg(r) +
\scanAlg(r+s))$, and the depth is at most $O(\log^2 (r+s))$.
For Step 2, the cost of running \procName{rankAll} is $O(\sortAlg(s)) =
O(\tfrac{s}{\BB}\log_{\MM/\BB}(1 + s/\BB))$ memory cost and $O(\log^2 s)$ depth.
The cost of the other operations in Step 2 is dominated by $O(\sortAlg(r) +
\sortAlg(s))$. The total cost for Step 2 is $O(\sortAlg(r) + \sortAlg(s))$
memory cost and at most $O(\log^2 (r+s))$ depth.  Likewise, in Step 2, the cost
of sorting dominates the multisearch cost and \procName{map}, resulting in a
total cost of $O(\sortAlg(r) + \sortAlg(s))$ for memory cost and at most
$O(\log^2 (r+s))$ depth.

Hence, the total cost of the \procName{bulkUpdateAll} is
$O(\sortAlg(r)+\sortAlg(s))$ memory cost and $O(\log^2 (r+s))$ depth, as promised.




\section{Evaluation}
\label{sec:exp}

We implemented the algorithm described in Section~\ref{sec:par-tri-count} and
investigated its performance on real-world datasets in terms of accuracy,
parallel speedup, parallelization overhead, effects of batch size on throughput,
and cache behaviors.






\medskip

\noindent\textbf{Implementation.} We followed the description in
Section~\ref{sec:par-tri-count} rather closely.  The algorithm combines the
``coarse'' estimators into a sharp estimate using a median-of-means aggregate.
The main optimization we made was in avoiding malloc-ing small chunks of memory
often.  The \procName{sort} primitive uses a PCO sample sort
algorithm~\cite{BGS10,ShunEFGKST:spaa2012}, which offers good speedups. The
multisearch routines are a modified Blelloch~\etal's \procName{merge}
algorithm~\cite{BGS10}; the modification stops recursing early when the number
of ``queries'' is small.  Other primitives are standard. The main triangle
counting logic has about $600$ lines of Cilk code, a dialect of C/C++ with
keywords that allow users to specify what should be run in
parallel~\cite{IntelCilk} (fork/join-type code).  Cilk's runtime system relies
on a work-stealing scheduler, a dynamic scheduler that allows tasks to be
rescheduled dynamically at relatively low cost. The scheduler is known to impose
only little overhead on both parallel and sequential code.
Our benchmark programs were compiled with the public version of Cilk shipped
with GNU \texttt{g++} Compiler version 4.8.0 (20130109). We used the
optimization flag \texttt{-O2}.

\medskip

\noindent\textbf{Testbed and Datasets.}
We performed experiments on a $12$-core (with hyperthreading) Intel machine,
running Linux 2.6.32-279 (CentOS 6.3).  The machine has \emph{two} $2.67$ Ghz
$6$-core Xeon X5650 processors with $96$GB of memory although the experiments
never need more a few gigabytes of memory.

Our study uses a collection of graphs, obtained from 
the SNAP project at Stanford~\cite{snap} and a recent Twitter
dataset~\cite{KwakLPM:www10}.
We present a summary of these datasets in Table~\ref{tbl:datastat}. 
We simulate a streaming graph by feeding the algorithm with edges as they are
read from disk. As we note later on, disk I/O is not a bottleneck in the
experiment.

With the growth of social media in recent years, some of the biggest graphs to
date arise in this context. When we went looking for big graphs to experiment
with, most graphs that turned up happen to stem from this domain.  While this
means we experimented with the graphs that triangle counting is likely used for,
we will note that our algorithm does not assume any special property about it.


\begin{table*}[t]
\centering
\small
\begin{tabular}{l r r r r r r r}
  \toprule
  \textbf{\textsf{Dataset}} & $n$         & $m$           & $\Delta$  & $\tau$          & $m\Delta/\tau$ & Size \\
  \midrule
  Amazon                    & 334,863     & 925,872       & 1,098     & 667,129         & 1,523.85       & 13M  \\
  DBLP                      & 317,080     & 1,049,866     & 686       & 2,224,385       & 323.78         & 14M  \\
  LiveJournal               & 3,997,962   & 34,681,189    & 29,630    & 177,820,130     & 5,778.89       & 0.5G \\
  Orkut                     & 3,072,441   & 117,185,083   & 66,626    & 627,584,181     & 12,440.68      & 1.7G \\
  Friendster                & 65,608,366  & 1,806,067,135 & 5,214     & 4,173,724,142   & 2,256.22       & 31G  \\
  \midrule
  Powerlaw (synthetic) & 267,266,082 &  9,326,785,184 &  6,366,528 & - & - & 167GB\\
  \bottomrule
\end{tabular}
\caption{A summary of the datasets used in our experiments,
  showing for every dataset, the number of nodes ($n$), the number of edges ($m$),
  the maximum degree ($\Delta$), the number of triangles in the graph ($\tau$), the ratio $m\Delta/\tau$, and
  size on disk (stored as a list of edges in plain text).}
\label{tbl:datastat}
\end{table*}


For most datasets, the exact triangle count is provided by the source (which we
have verified); in other cases, we compute the exact count using an algorithm
developed as part of the Problem-Based Benchmark
Suite~\cite{ShunEFGKST:spaa2012}.  We also report the size on disk of these
datasets in the format we obtain them (a list of edges in plain
text)\footnote{While storing graphs in, for example, the compressed sparse-row
  (CSR) format can result in a smaller footprint, this set of experiments
  focuses on the settings where we do not have the luxury of preprocessing the
  graph.}.
In addition, we include one synthetic power-law graph; on this graph, we cannot
obtain the true count, but it is added to speed test the algorithm.






\medskip

\noindent\textbf{Baseline.} For accuracy study, we directly compare our results
with the true count.  For performance study, our baseline is a fairly optimized
version of the nearly-linear time algorithm based on neighborhood sampling, using
bulk-update, as described in~\cite{PavanTTW:vldb2013}. We use this baseline to establish the overhead of
the parallel algorithm. 
We do not compare the accuracy between the two algorithms because by design, they produce
the exact same answer given the same sequence of random bits.  The baseline code
was also compiled with \texttt{g++} version 4.8.0 \emph{but} without linking
with the Cilk runtime.

\subsection{Performance and Accuracy}
We perform experiments on graphs with varying sizes and densities.
Our algorithm is randomized and may behave differently on different runs.  For
robustness, we perform \emph{five} trials---except when running the biggest
datasets on a single core, where only two trials are used.
Table~\ref{tbl:overview-perf} shows for different numbers of estimators $r =
200\text{K}, 2\text{M}, 20\text{M}$, the accuracy, reported as the mean
deviation value, and processing times (excluding I/O) using $1$ and all $12$
cores (24 threads with hyperthreading), as well as the speedup ratio\footnote{A
  batch size of $16$M was used}.
%
Mean deviation is a well-accepted measure of error, which, we believe,
accurately depicts how well the algorithm performs.  In addition, it reports the
median I/O time\footnote{Like in the (streaming) model, the update routine to
  our algorithm takes in a batch of edges, represented as an array of a pair of
  \procName{int}'s. We note that the I/O reported is based on an optimized I/O
  routine, in place of the \procName{fstream}'s \procName{cin}-like
  implementation or \procName{scanf}.}  In another experiment, presented in
Table~\ref{tbl:par-overhead}, we compare our parallel implementation with the
baseline sequential implementation~\cite{PavanTTW:vldb2013}.  As is evident
from the findings, \emph{the I/O cost is not a bottleneck in any of the
  experiments, justifying the need for parallelism to improve throughput.}

\begin{table*}[tbh]
\centering
\small
\renewcommand{\tabcolsep}{2pt}
\begin{tabular}{l r r r r c  r r r r c r r r r c r}
  \toprule
  \multirow{2}{*}{\textbf{\textsf{Dataset}}} &
  \multicolumn{4}{c}{$r=200$K} & &
  \multicolumn{4}{c}{$r=2$M} & &
  \multicolumn{4}{c}{$r=20$M} & &
  \multirow{2}{*}{I/O}\\[0.1em]
  \cmidrule{2-5}
  \cmidrule{7-10}
  \cmidrule{12-15}
  \mbox{}&
  MD & $T_1$ & $T_{12H}$ & $\frac{T_1}{T_{12H}}$ &\tspc&
  MD & $T_1$ & $T_{12H}$ & $\frac{T_1}{T_{12H}}$ &\tspc&
  MD & $T_1$ & $T_{12H}$ & $\frac{T_1}{T_{12H}}$\\

  \midrule
  Amazon & $2.08$ & $1.07$ &$0.14$ &  $7.64$ & & $0.47$ & $3.58$ &$0.42$ &  $8.52$ & & $0.12$ & $29.00$ &$3.15$ &  $9.21$ & & 0.14 \\
DBLP & $1.18$ & $1.21$ &$0.17$ &  $7.12$ & & $0.39$ & $3.71$ &$0.43$ &  $8.63$ & & $0.08$ & $29.20$ &$3.14$ &  $9.30$ & & 0.15 \\
LiveJournal & $5.51$ & $35.10$ &$3.31$ &  $10.60$ & & $0.51$ & $41.00$ &$3.91$ &  $10.49$ & & $0.39$ & $69.80$ &$6.56$ &  $10.64$ & & 1.52 \\
Orkut & $2.39$ & $119.00$ &$11.20$ &  $10.62$ & & $0.16$ & $133.00$ &$12.40$ &  $10.73$ & & $0.66$ & $172.00$ &$15.40$ &  $11.17$ & & 5.00 \\
Friendster & $18.25$ & $1730.00$ &$181.00$ &  $9.56$ & & $8.35$ & $1970.00$ &$193.00$ &  $10.21$ & & $3.41$ & $2330.00$ &$219.00$ &  $10.64$ & & 86.00 \\

  \midrule
  Powerlaw & $-$ & $-$ &$-$ &  $-$ & & $-$ & $-$ &$-$ &  $-$ & & $-$ & $-$ &$1050.0$ &  $-$ & & $970.00$ \\
  \bottomrule
\end{tabular}
\caption{The accuracy (MD is mean deviation, \textbf{in percentage}), median processing time on $1$ core $T_1$ (\textbf{in seconds}), median processing time on $12$ cores $T_{12H}$ with hyperthreading (\textbf{in seconds}), and I/O time (\textbf{in seconds}) of our parallel  algorithm across five runs as the number of estimators $r$ is varied.}
\label{tbl:overview-perf}
\end{table*}


Several trends are evident from these experiments.  \emph{First, the algorithm
  is accurate with only a modest number of estimators.}  In all datasets,
including the one with more than a billion edges, the algorithm achieves less
than $4$\% mean deviation using about $20$ million estimators, and for smaller
datasets, it can obtain better than $5$\% mean deviation using fewer estimators.
As a general trend---though not a strict pattern---the accuracy improves with
the number of estimators.  Furthermore, in practice, far fewer estimators than
suggested by the pessimistic theoretical bound is necessary to reach a desired
accuracy. For example, on Twitter-2010, which has the highest $m\Delta/\tau$
ratio among the datasets, using $\vareps=0.0386$, the expression $96/\vareps^2
\cdot m\Delta/\tau$ (see Theorem~\ref{thm:mdelta-suffices}) is at least $6.6$
billion, but we reach this accuracy using $20$ million estimators.

\emph{Second, the algorithm shows substantial speedups on all datasets.}  On all
datasets, the experiments show that the algorithm achieves up to $11.24$x
speedup on $12$ cores, with the speedup numbers ranging between $7.5$x and
$11.24$x.  On the biggest datasets using $r = 20\text{M}$ estimators, the
speedups are consistently above $10$x.  This shows that the algorithm is able to
effectively utilize available cores except when the dataset or the number of
estimators is too small to fully utilize parallelism.
%
%
Additionally, we experimented with a big synthetic graph (167GB power-law graph)
to get a sense of the running time.  For this dataset, we were unable to
calculate the true count; we also cut short the sequential experiment after a
few hours. But it is worth pointing out that this dataset has $5$x more edges
than Friendster, and our algorithm running on 12 cores finishes in $1050$
seconds (excluding I/O)---about $5$x longer than on Friendster.

\begin{table*}[ht]
\centering
\small
\renewcommand{\tabcolsep}{2pt}
\begin{tabular}{l r r r c  r r r c r r r c}
  \toprule
  \multirow{2}{*}{\textbf{\textsf{Dataset}}} &
  \multicolumn{3}{c}{$r=200$K} & &
  \multicolumn{3}{c}{$r=2$M} & &
  \multicolumn{3}{c}{$r=20$M} \\[0.1em]
  \cmidrule{2-4}
  \cmidrule{6-8}
  \cmidrule{10-12}
  \mbox{}&
  $T_{\textrm{seq}}$ & $T_{1}$ & $\sfrac{T_1}{T_{\textrm{seq}}}$ &\tspc&
  $T_{\textrm{seq}}$ & $T_{1}$ & $\sfrac{T_1}{T_{\textrm{seq}}}$ &\tspc&
  $T_{\textrm{seq}}$ & $T_{1}$ & $\sfrac{T_1}{T_{\textrm{seq}}}$ \\
  \midrule
  Amazon & $0.95$ & $1.07$ &$1.13$ & & $5.25$ & $3.58$ &$0.68$ & & $34.80$ & $29.00$ &$0.83$ & \\
DBLP & $0.92$ & $1.21$ &$1.32$ & & $5.16$ & $3.71$ &$0.72$ & & $36.90$ & $29.20$ &$0.79$ & \\
LiveJournal & $14.00$ & $35.10$ &$2.51$ & & $34.90$ & $41.00$ &$1.17$ & & $84.40$ & $69.80$ &$0.83$ & \\
Orkut & $42.00$ & $119.00$ &$2.83$ & & $73.20$ & $133.00$ &$1.82$ & & $129.00$ & $172.00$ &$1.33$ & \\
Friendster & $787.00$ & $1730.00$ &$2.20$ & & $1040.00$ & $1970.00$ &$1.89$ & & $1580.00$ & $2330.00$ &$1.47$ & \\

  \bottomrule
\end{tabular}
\caption{The median processing time of the sequential algorithm $T_{\text{seq}}$ (\textbf{in seconds}), 
  the median processing time of the parallel algorithm running on $1$ core $T_1$ (\textbf{in seconds}), 
  and the overhead factor (i.e., $T_1/T_{\text{seq}}$).}
\label{tbl:par-overhead}
\end{table*}


\emph{Third, the overhead is well-controlled.}  Both the sequential and parallel
algorithms, at the core, maintain the same neighborhood sampling invariant but
differ significantly in how the edges are processed. 
As is apparent from Table~\ref{tbl:par-overhead}, for large datasets
requiring more estimators, the overhead is less than $1.5$x with $r =
20\text{M}$.  For smaller datasets, the overhead is less than $1.6$x with $r =
2\text{M}$. In all cases, the amount of speedup gained outweighs the overhead.

\subsection{Additional Experiments}

The previous set of experiments studies the accuracy, scalability, and overhead
of our parallel algorithm, showing that it scales well and has low overhead.  We
report on another set of experiments designed to help understand the algorithm's
behaviors in more detail.

\medskip

\begin{figure}[thb]
  \centering
  \includegraphics[width=.5\columnwidth]{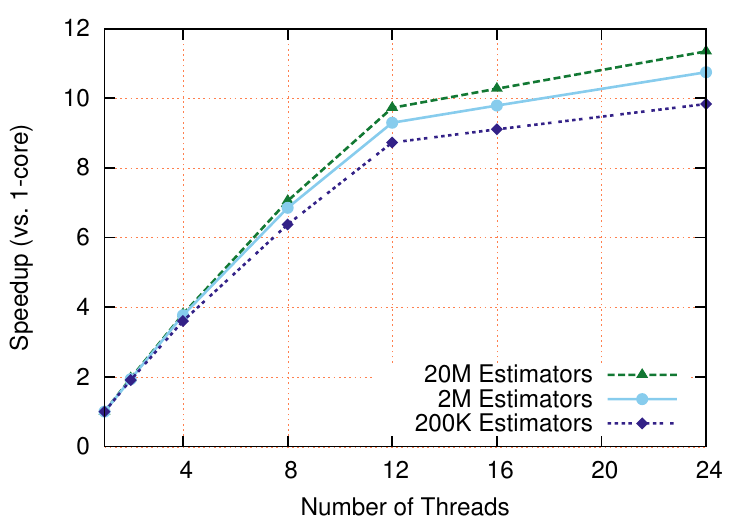}
	\caption{Speedup on the Friendster dataset as the number of cores and the number of estimators are varied.}
  \label{fig:speedup-vs-core}
\end{figure}

\noindent\textbf{Scalability.} To further our understanding of the algorithm's
scalability, we take the biggest real dataset (Friendster) and examine its
speedup behavior as the number of threads is varied.
Figure~\ref{fig:speedup-vs-core} shows, for $r = 200\mbox{K}, 2\mbox{M},
20\mbox{M}$, the speedup ratio with respect to the parallel code running on $1$
core. On a $12$-core machine, the speedup grows linearly with the number of
threads until $12$ threads (the number of physical cores); the growth slows down
after that but keeps rising until $24$ threads since we are now running on
hyperthreaded cores.  The trend seems invariable with the number of estimators
$r$ although the speedup ratio improves with larger $r$ since there is more work
for parallelism.

As a next step, we attempt to understand how the speedup is gained and how it is
likely to scale in the future on other machines.  For this, we examine the
breakdown of how the time is spent in different components of the algorithm.
Figure~\ref{fig:breakdown} shows the fractions of time spent inside
\procName{sort}, multisearch routines, and other components for the $3$ biggest
datasets. 
The figure shows that the majority of the time is spent on sorting (up to
$94\%$), which has great speedups in our setting. The multisearch portion makes
up less than $5\%$ of the running time. The remaining $1\%$ is spent in
miscellaneous bookkeeping.  Because we can directly use an off-the-shelf sorting
implementation, this main portion will improve with the performance of
\procName{sort}.  We expect the algorithm to continue to scale well on machines
with more cores, as cache-efficient sort has been shown to
scale~\cite{ShunEFGKST:spaa2012}.


\begin{figure}[tb]
  \centering
  \includegraphics[width=0.65\columnwidth]{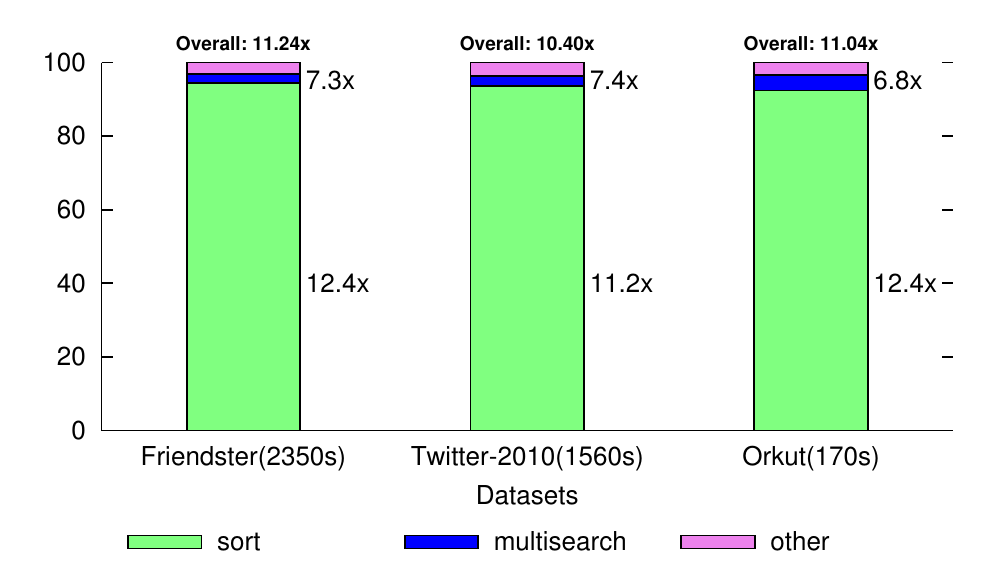}
  \caption{The timing breakdown (\textbf{in percentage of time}) of the parallel algorithm on the $3$ largest
    datasets running on $1$ core.  The number in parenthesis next to each
    dataset is the $T_1$ (\textbf{in seconds}) whereas the numbers next to the
    columns are the speedup ratio for the corresponding components running with
    24 threads (all 12 cores with hyperthreading).} \label{fig:breakdown}
\end{figure}

\medskip
\noindent
\textbf{Batch Size vs. Throughput.}  How does the choice of a batch size affect
the overall performance of our algorithm?  We study this question empirically by
measuring the sustained throughput as the batch size is varied.
Figure~\ref{fig:xput-vs-batchsize} shows the results on the dataset Orkut
running with $r = 2\mbox{M}$ and $24$ threads.  Orkut was chosen so that we
understand the performance trend on medium-sized graphs, showing that our scheme
can benefit not only the largest graphs but also smaller ones.  We expect
intuitively that the throughput improves with the batch size, and the experiment
confirms that, reaching about $9.3$ million edges per second using a batch size
of $16\mbox{M}$.
\begin{figure}[th]
  \centering
  \includegraphics[width=0.5\columnwidth]{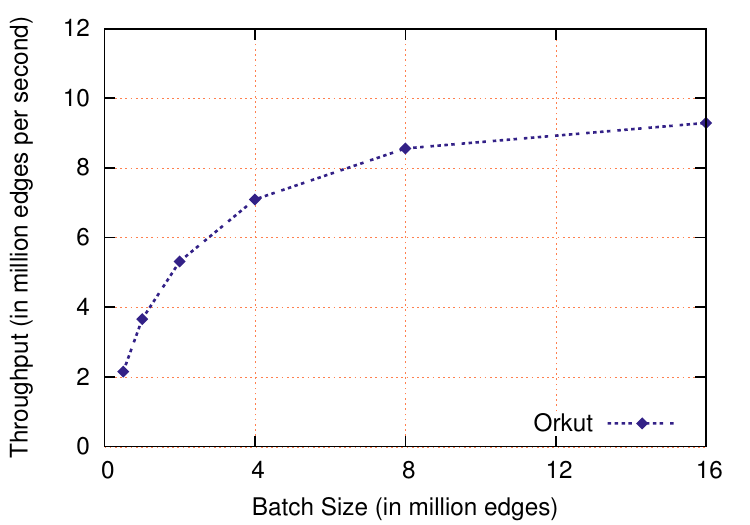}
  \caption{The sustained throughput the algorithm with varying batch sizes on
    the graph Orkut using $r = 2\mbox{M}$ and $24$ threads.}
  \label{fig:xput-vs-batchsize}
\end{figure}

\medskip

\begin{figure}[th]
  \centering
  \includegraphics[width=0.5\columnwidth]{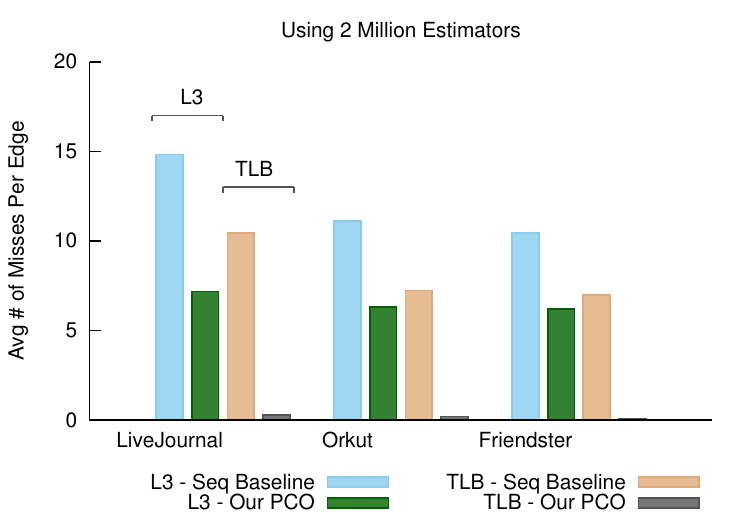}
  \caption{The average number of L3 cache misses per edge and the average number
    of TLB misses per edge as reported by \texttt{perftool} on representative
    datasets running with $r = 2$M estimators. Both algorithms were run
    single-threaded.}
  \label{fig:faults-count}
\end{figure}

\noindent\textbf{Cache Behavior.}
We also study the memory access cost of the parallel cache-oblivious algorithm
compared to the sequential baseline, which uses hash tables and performs a
number of random accesses.  Both algorithms were run on a single core with the
same number of estimators (hence roughly the same memory footprint).  For this
study, we consider the number of L3 misses and TLB (Translation Lookaside
Buffer) misses.  Misses in these places are expensive, exactly what a
cache-efficient algorithm tries to minimize.  In Figure~\ref{fig:faults-count},
we show the numbers of L3 and TLB misses normalized by the number of edges
processed; the normalization helps compares results from datasets of different
sizes.  From this figure, it is clear that our cache-efficient scheme has
significantly lower cache/TLB misses.

\section{Conclusion}
\label{sec:concl}

We have presented a cache-efficient parallel algorithm for approximating the
number of triangles in a massive streaming graph.  The proposed algorithm is
cache-oblivious and has good theoretical performance and accuracy guarantees. We
also experimentally showed that the algorithm is fast and accurate, and has low
cache complexity. It will be interesting to explore other problems at the
intersection of streaming algorithms and parallel processing.


{\small
 \subsubsection*{Acknowledgments}
 Tangwongsan was in part sponsored by the U.S. Defense Advanced Research Projects
Agency (DARPA) under the Social Media in Strategic Communication (SMISC)
program, Agreement Number W911NF-12-C-0028. The views and conclusions contained
in this document are those of the author(s) and should not be interpreted as
representing the official policies, either expressed or implied, of the
U.S. Defense Advanced Research Projects Agency or the U.S. Government. The
U.S. Government is authorized to reproduce and distribute reprints for
Government purposes notwithstanding any copyright notation hereon. Pavan was
sponsored in part by NSF grant 0916797; Tirthapura was sponsored in part by
NSF grants 0834743 and 0831903.

}

\setlength{\bibsep}{2pt}

{
\bibliographystyle{alpha}
\bibliography{thispaper}
}

\appendix
\section{More Detailed  Models}
\label{sec:wd-pco}

Parallel algorithms in this work are expressed in the nested parallel model.  It
allows arbitrary dynamic nesting of parallel loops and fork-join constructs but
no other synchronizations, corresponding to the class of algorithms with
series-parallel dependency graphs. 
In this model, computations can be recursively decomposed into tasks, parallel
blocks, and strands, where the top-level computation is always a task:
\begin{itemize}[itemsep=1pt,topsep=1pt]
\item The smallest unit is a \defn{strand} $\Qs$, a serial sequence of instructions
  not containing any parallel constructs or subtasks.
\item A \defn{task} $\Qt$ is formed by serially composing $k \geq 1$ strands
  interleaved with $k-1$ parallel blocks, denoted by $\Qt =
  \Qs_1;\Qb_1;\dots;\Qs_k$.
\item A \defn{parallel block} $\Qb$ is formed by composing in parallel one or
  more tasks with a fork point before all of them and a join point after,
  denoted by $\Qb = \Qt_1\|\Qt_2\|\dots\|\Qt_k$. A parallel block can be, for
  example, a parallel loop or some constant number of recursive calls.
\end{itemize}
The \defn{depth} (aka. \defn{span}) of a computation is the length of the
longest path in the dependence graph.


\begin{figure}[th]
  \centering
  \includegraphics[scale=0.75]{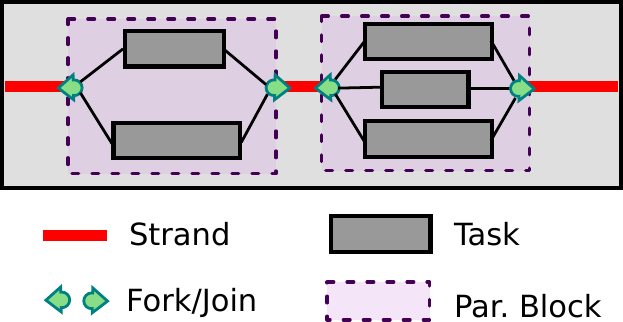}
  \caption{Expressing a computation as tasks, strands and parallel
    blocks}\label{fig:tasks}
\end{figure}


%
We analyze memory-access cost of parallel algorithms in the Parallel Cache
Oblvivious (PCO) model~\cite{BFGS11}, a parallel variant of the cache oblivious
(CO) model.  The Cache Oblivious (CO) model~\cite{FLPR99} is a model for
measuring cache misses of an algorithm when run on a single processor machine
with a two-level memory hierarchy---one level of finite cache and unbounded
memory.  The cache complexity measure of an algorithm under this model $\CoQ{n;
  \MM{},\BB{}}$ counts the number of cache misses incurred by a problem instance
of size $n$ when run on a fully associative cache of size $\MM{}$ and line size
$\BB{}$ using the optimal (offline) cache replacement policy.

Extending the CO model, the PCO model gives a way to analyze the number of cache
misses for the tasks that run in parallel in a parallel block.  The PCO model
approaches it by (i) ignoring any data reuse among the parallel subtasks and
(ii) assuming the cache is flushed at each fork and join point of any task that
does not fit within the cache.

More precisely, let $\loc(\Qt{};\BB{})$ denote the set of distinct cache lines
accessed by task $\Qt{}$, and $S(\Qt{};\BB{}) = |\loc(\Qt{};\BB{})| \cdot \BB{}$
denote its size. Also, let $s(\Qt; \BB) = |\loc(\Qt;\BB)|$ denote the size in
terms of number of cache lines.  Let $Q(\Qc{};\MM{},\BB{};\kappa)$ be the cache
complexity of $\Qc{}$ in the sequential CO model when starting with cache state
$\kappa$.

\begin{definition}[Parallel Cache-Oblivious Model]\label{def:qstar}
  For cache parameters $\MM$ and $\BB$ the \defn{cache complexity} of a strand,
  parallel block, and a task starting in cache state $\kappa$ are defined
  recursively as follows (see~\cite{BFGS11} for detail).
  \begin{itemize}[topsep=2pt]
  \item For a strand, $\cc(\Qs{};\MM{},\BB{};\kappa) =
    Q(\Qs{};\MM{},\BB{};\kappa)$.

  \item For a parallel block $\Qb{} = \Qt_1\|\Qt_2\|\ldots\|\Qt_k$,
    $\cc(\Qb{};\MM{},\BB{};\kappa) = \sum_{i=1}^{k}
    \cc(\Qt{}_i;\MM{},\BB{};\kappa)$.

  \item
    For a task $\Qt = \Qc_1;\dots;\Qc_k$,\\
    $\cc(\Qt{};\MM{},\BB{};\kappa)$ $= \sum_{i=1}^{k}\cc(\Qc_i;\MM{},\BB{};\kappa_{i-1})$,
    where $\kappa_i = \emptyset$
    if $S(\Qt;\BB) >\MM$, and $\kappa_i = \kappa\cup_{j=1}^{i}\loc(\Qc_j;\BB)$ if
    $S(\Qt;\BB) \leq \MM$.
  \end{itemize}
\end{definition}

We use $\cc(\Qc{};\MM{},\BB{})$ to denote a computation $\Qc{}$ starting with an
empty cache and overloading notation, we write $\cc(n;\MM{},\BB{})$ when $n$ is a
parameter of the computation.
We note that $\cc(\Qc{};\MM{},\BB{}) \geq Q(\Qc{};\MM{},\BB{})$.  That is, the
PCO gives cache complexity costs that are always at least as large as the CO
model.  Therefore, any upper bound on the PCO is an upper bound on the CO model.
Finally, when applied to a parallel machine, $\cc$ is a ``work-like'' measure
and represents the total number of cache misses across all processors.  An
appropriate scheduler is used to evenly balance them across the processors.

\section{Scan With Resets}
\label{sec:scan-with-resets}
The following technique is standard (see, e.g., \cite{Jaja:book92}); we only
present it here for reference. We describe how to implement a parallel prefix
sum operation with reset (also known as segmented scan), where the input is a
sequence $A$ such that $A_i \in \{\textsf{1}, \emptyset\}$. Sequentially, the
desired output can computed by the following code:
\begin{quote}
\begin{mlcodea}
sum = 0\\
for i = 1 to |A| do\\
\>\>out[i] = sum\\
\>\>if (A[i] == $\emptyset$) sum = 0 else sum += A[i]\\
end\\
\end{mlcodea}
\end{quote}
\vspace{-0.3cm} 
That is, we keep a accumulator which is incremented every time a $\textsf{1}$ is
encountered and reset back to $0$ upon encountering $\emptyset$.  This process
is easy to parallelize using \procName{scan} with a binary associative operator
such as the following. Let $\oplus: D \times D \to D$, where $D = \Z \cup
\{\emptyset\}$, be given by
\begin{align*}
  \oplus(x, y) = \begin{cases}
    x + y & \text{ if } x \neq \emptyset \text{ and } y \neq \emptyset\\
    y & \text{ if } y \neq \emptyset\\
    0 & \text{ otherwise }
  \end{cases}
\end{align*}



\end{document}